\documentclass[12pt]{article}
\newenvironment{proof}
{\pagebreak[1]{\narrower\noindent {\bf Proof:\quad\nopagebreak}}}{\QED}

\usepackage{amsmath}
\usepackage{amssymb}
\usepackage{url}
\usepackage{graphics}
\usepackage{enumerate}
\usepackage{turnstile}
\usepackage{centernot}

\newcommand{\ang}[1]{\langle#1\rangle}

\newcommand{\xvec}[1]{\ifcase 3{#1} {\ang {x_1,x_2,x_3} } \else
\ifcase 4{#1} {\ang{x_1,x_2,x_3,x_4}} \else {\ang {x_1,\ldots,x_{#1}}}\fi\fi}
\newcommand{\yvec}[1]{\ifcase 3{#1} {\ang {y_1,y_2,y_3} } \else
\ifcase 4{#1} {\ang{y_1,y_2,y_3,y_4}} \else {\ang {y_1,\ldots,y_{#1}}}\fi\fi}
\newcommand{\zvec}[1]{\ifcase 3{#1} {\ang {z_1,z_2,z_3} } \else
\ifcase 4{#1} {\ang{z_1,z_2,z_3,z_4}} \else {\ang {z_1,\ldots,z_{#1}}}\fi\fi}
\newcommand{\vecc}[2]{\ifcase 3{#2} {\ang { {#1}_1,{#1}_2,{#1}_3 } } \else
\ifcase 4{#1} {\ang { {#1}_1,{#1}_2,{#1}_3,{#1}_{4} } }
\else {\ang { {#1}_1,\ldots,{#1}_{#2}}}\fi\fi}
\newcommand{\veccd}[3]{\ifcase 3{#2} {\ang { {#1}_{{#3}1},{#1}_{{#3}2},{#1}_{{#3}3} } } \else
\ifcase 4{#1} {\ang { {#1}_{{#3}1},{#1}_{{#3}2},{#1}_{#3}3},{#1}_{{#3}4} }
\else {\ang { {#1}_{{#3}1},\ldots,{#1}_{{#3}{#2}}}}\fi\fi}
\newcommand{\veccz}[2]{\ifcase 3{#2} {\ang { {#1}_0,{#1}_2,{#1}_3 } } \else
\ifcase 4{#1} {\ang { {#1}_0,{#1}_2,{#1}_3,{#1}_{4} } }
\else {\ang { {#1}_0,\ldots,{#1}_{#2}}}\fi\fi}
\newcommand{\xve}[1]{\ifcase 3{#1} {x_1,x_2,x_3} \else
\ifcase 4{#1} {x_1,x_2,x_3,x_4} \else {x_1,\ldots,x_{#1}}\fi\fi}
\newcommand{\yve}[1]{\ifcase 3{#1} {y_1,y_2,y_3} \else
\ifcase 4{#1} {y_1,y_2,y_3,y_4} \else {y_1,\ldots,y_{#1}}\fi\fi}
\newcommand{\zve}[1]{\ifcase 3{#1} {z_1,z_2,z_3} \else
\ifcase 4{#1} {z_1,z_2,z_3,z_4} \else {z_1,\ldots,z_{#1}}\fi\fi}
\newcommand{\ve}[2]{\ifcase 3#2 {{#1}_1,{#1}_2,{#1}_3} \else
\ifcase 4#2 {{#1}_1,{#1}_2,{#1}_3,{#1}_{4}}
\else {{#1}_1,\ldots,{#1}_{#2}}\fi\fi}
\newcommand{\ved}[3]{\ifcase 3#2 {{#1}_{{#3}1},{#1}_{{#3}2},{#1}_{{#3}3}} \else
\ifcase 4#2 {{#1}_{{#3}1},{#1}_{{#3}2},{#1}_{{#3}3},{#1}_{{#3}4}}
\else {{#1}_{{#3}1},\ldots,{#1}_{{#3}{#2}}}\fi\fi}
\newcommand{\fuve}[3]{
\ifcase 3#2
{{#3}({#1}_1),{#3}({#1}_2,{#3}({#1}_3)} \else
\ifcase 4#2
{{#3}({#1}_1),{#3}({#1}_2),{#3}({#1}_3),{#3}({#1}_4)}
\else
{{#3}({#1}_1),\ldots,{#3}({#1}_{#2})}\fi\fi}
\newcommand{\setmathchar}[1]{\ifmmode#1\else$#1$\fi}
\newcommand{\vlist}[2]{%
	\setmathchar{%
		\compound#2\one{#2}\two
		\ifcompound
			({#1}_1,\ldots,{#1}_{#2})
		\else
			\ifcat N#2
				({#1}_1,\ldots,{#1}_{#2})
			\else
				\ifcase#2
					({#1}_0)\or
					({#1}_1)\or
					({#1}_1,{#1}_2)\or
					({#1}_1,{#1}_2,{#1}_3)\or
					({#1}_1,{#1}_2,{#1}_3,{#1}_4)\else
					({#1}_1,\ldots,{#1}_{#2})
				\fi
			\fi
		\fi}}

\newif\ifcompound
\def\compound#1\one#2\two{%
	\def\one{#1}
	\def\two{#2}
	\if\one\two
		\compoundfalse
	\else
		\compoundtrue
	\fi}

\newcommand{\xwe}[1]{\ifcase 3{#1} {x_1\wedge x_2\wedge x_3} \else
\ifcase 4{#1} {x_1\wedge x_2\wedge x_3\wedge x_4} \else {x_1\wedge\citedots \wedge
x_{#1}}\fi\fi}
\newcommand{\we}[2]{\ifcase 3#2 {\ang { {#1}_1\wedge {#1}_2\wedge {#1}_3 } } \else
\ifcase 4{#1} {\ang { {#1}_1\wedge {#1}_2\wedge {#1}_3\wedge {#1}_{4} } }
\else {\ang { {#1}_1\wedge\citedots\wedge {#1}_{#2}}}\fi\fi}
\newcommand{\s}[1]{\s_{#1}}

\newcommand{\monus}{\;\raise.5ex\hbox{{${\buildrel
    \ldotp\over{\hbox to 6pt{\hrulefill}}}$}}\;}

\newcounter{savenumi}

\newtheorem{theoremfoo}{Theorem}[section] %by chapter in report style
\newenvironment{theorem}{\pagebreak[1]\begin{theoremfoo}}{\end{theoremfoo}}

\newtheorem{lemmafoo}[theoremfoo]{Lemma}
\newenvironment{lemma}{\pagebreak[1]\begin{lemmafoo}}{\end{lemmafoo}}
\newtheorem{conjecturefoo}[theoremfoo]{Conjecture}

\newtheorem{conventionfoo}[theoremfoo]{Convention}

\newtheorem{porismfoo}[theoremfoo]{Porism}

\newtheorem{gamefoo}[theoremfoo]{Game}

\newtheorem{corollaryfoo}[theoremfoo]{Corollary}

\newtheorem{claimfoo}[theoremfoo]{Claim}

\newtheorem{openfoo}[theoremfoo]{Open Problem}

\newtheorem{exercisefoo}{Exercise}

\newcommand{\fig}[1] %usage:\fig{file}
{
 \begin{figure}
 \begin{center}
 \input{#1}
 \end{center}
 \end{figure}
}

\newtheorem{potanafoo}[theoremfoo]{Potential Analogue}

\newtheorem{notefoo}[theoremfoo]{Note}

\newtheorem{notabenefoo}[theoremfoo]{Nota Bene}

\newtheorem{nttn}[theoremfoo]{Notation}

\newtheorem{empttn}[theoremfoo]{Empirical Note}

\newtheorem{examfoo}[theoremfoo]{Example}

\newtheorem{dfntn}[theoremfoo]{Definition}

\newtheorem{propositionfoo}[theoremfoo]{Proposition}

\newcommand{\yyskip}{\penalty-50\vskip 5pt plus 3pt minus 2pt}
\newcommand{\blackslug}{\hbox{\hskip 1pt
        \vrule width 4pt height 8pt depth 1.5pt\hskip 1pt}}
\newcommand{\QED}{{\penalty10000\parindent 0pt\penalty10000
        \hskip 8 pt\nolinebreak\blackslug\hfill\lower 8.5pt\null}
        \par\yyskip\pagebreak[1]}

\newcommand{\BBB}{{\penalty10000\parindent 0pt\penalty10000
        \hskip 8 pt\nolinebreak\hbox{\ }\hfill\lower 8.5pt\null}
        \par\yyskip\pagebreak[1]}

\newtheorem{factfoo}[theoremfoo]{Fact}

\newenvironment{block}{\begin{list}{\hbox{}}{\leftmargin 1em
    \itemindent -1em \topsep 0pt \itemsep 0pt \partopsep 0pt}}{\end{list}}

\usepackage{enumitem}
\usepackage{xcolor}
\usepackage[all]{xy}
\usepackage{wrapfig}

\title{Ruling Out Short Proofs of Unprovable Sentences is Hard}
\author{Hunter Monroe}
\date{\today}
\begin{document}
\maketitle
\begin{abstract}
If no optimal propositional proof system exists, we (and independently Pudlák) prove that ruling out length $t$ proofs of any unprovable sentence is hard. This mapping from unprovable to hard-to-prove sentences powerfully translates facts about noncomputability into complexity theory. For instance, because proving string $x$ is Kolmogorov random ($x{\in}R$) is typically impossible, it is typically hard to prove ``no length $t$ proof shows $x{\in}R$'', or tautologies encoding this. Therefore, a proof system with one family of hard tautologies has these densely in an enumeration of families. The assumption also implies that a natural language is $\textbf{NP}$-intermediate: with $R$ redefined to have a sparse complement, the complement of the language $\{\langle x,1^t\rangle|$ no length $t$ proof exists of  $x{\in}R\}$ is also sparse. 

Efficiently ruling out length $t$ proofs of $x{\in}R$ might violate the constraint on using the fact of $x{\in}R$'s unprovability. We conjecture: any computable predicate on $R$ that might be used in if-then statements (or case-based proofs) does no better than branching at random, because $R$ appears random by any effective test. This constraint could also inhibit the usefulness in circuits and propositional proofs of NOT gates and cancellation---needed to encode if-then statements. If $R$ defeats if-then logic, exhaustive search is necessary.
\end{abstract}
\section{Introduction}
We prove a deep linkage between noncomputability and complexity under a widely believed conjecture---that there is no optimal propositional proof system for tautologies.\footnote{This paper was prepared in honor of past and present faculty of Davidson College, including Hansford Epes, L. Richardson King, Benjamin Klein, and Clark Ross. Comments are appreciated from Pavel Pudlák and Bill Gasarch. The ideas in this paper and earlier versions have benefited from discussions with the following: Scott Aaronson, Eric Allender, Olaf Beyersdorff, Ilario Bonacina, Maria Luisa Bonet, Cristian Calude, Marco Carmosino, Yuval Filmus, Vijay Ganesh, Bill Gasarch, Valentina Harizonov, Pavel Hrubeš, Rahul Ilango, Russell Impagliazzo, Valentine Kabanets, Mehmet Kayaalp, Yanyi Liu, Ian Mertz, Daniel Monroe, Igor Oliveira, Toniann Pitassi, Hanlin Ren, Rahul Santhanam, Till Tantau, Neil Thapen, Luca Trevisan, Avi Wigderson, Ryan Williams, Marius Zimand, and other participants in seminars at George Washington University and Davidson College, the Simons Institute 2023 Meta-Complexity Program, the Computational Complexity Conference 2022, the Workshop on Proof Complexity 2022, and the Conference on Complexity with a Human Face 2022. Remaining errors are my own.} That conjecture originated as an assertion that a noncomputability result also holds with a resource bound. Gödel’s Second Incompleteness Theorem states that no consistent sufficiently powerful theory can prove its own consistency. Pudlák\cite{Pudlak1986length} and Friedman independently formulated a feasible consistency conjecture: it is hard to rule out any length $t$ proof in a theory of its own inconsistency.\footnote{See Pudlák\cite{PudlakLogicalFoundations} Section 6.4 and \cite{PudlakFiniteDomain}. Pudlák\cite{Pudlak1986length} shows the initial conjecture was incorrect---a theory $\mathcal{T}$ can efficiently prove that $\mathcal{T}$ lacks a length $t$ proof of `$0{=}1$'. The 1989 reformulation refers to the lack of efficient proofs in a weaker theory. See also Theorem 59 of Pudlák\cite{PudlakLogicalFoundations}.}  Kraj\'{\i}\v{c}ek and Pudlák\cite{Krajicek} proved the lack of efficient proofs (in a weaker theory) of inconsistency  is equivalent to the nonexistence of an optimal proof system, which remains a key conjecture in proof complexity theory.\footnote{See also Kraj\'{\i}\v{c}ek\cite{Krajicekproof} Section 21.3.} 

We show: if it is possible to efficiently rule out length $t$ proofs of some unprovable sentence $\phi$, it is also possible to efficiently rule out a slightly shorter proof of inconsistency, which could be used in a length $t$ proof of $\phi$ by contradiction. This implies a powerful generalization---if it is hard to rule out length $t$ proofs of inconsistency, it is hard to rule of length $t$ proofs of any unprovable sentence. This in turn implies that facts about unprovability and noncomputability, which are well understood, can be imported into complexity theory. This has wide ramifications---diverse types of unprovable sentences translate into assertions that open questions in complexity theory have the expected answers. For instance, unprovable sentences of the form $x{\in}R$ are dense, so hard families of tautologies encoding ``no length $t$ proof shows $x{\in}R$'' are also dense. With $R$ redefined to have a sparse complement---a string is in $R$ unless exponentially compressible---the complement of the language $\{\langle x,1^t\rangle|$ no length $t$ proof exists of  $x{\in}R\}$ is neither in $\textbf{P}$ nor $\textbf{NP}$-complete, but is $\textbf{NP}$-intermediate. 

The hardness of ruling out length $t$ proofs of any unprovable sentence implies a deep linkage between noncomputability and complexity. We show that the implicit mapping from unprovable sentences to families of hard-to-prove sentences in a theory is an isomorphism. This would be a significant previously unnoticed structural feature of theories such as ZFC.  

Formalizing the intuition ``ruling out length $t$ proofs is hard'' requires specifying which theory lacks length $t$ proofs and which theory has difficulty ruling them out. These theories must be different, as a theory that proves it lacks short proofs of some $\phi$ would prove its own consistency. Our main result is:

\begin{theorem}\label{theoremmain}
The following are equivalent:\footnote{Monroe\cite{MonroeTCS} shows another equivalent condition: For any $M$ accepting $\texttt{coBHP}=\{\langle N,y,1^t\rangle|$ there is no accepting path of nondeterministic TM (NTM) $N$ on input $y$ with $t$ or fewer steps$\}$, there exists some $\langle N',y'\rangle$ where $N'$ does not halt on $y'$ such that $\langle N',y',1^t\rangle$ is a hard family of inputs.}

(i) No optimal propositional proof system exists.

(ii) For consistent theory $\mathcal{S}$, for some stronger theory $\mathcal{T}$, $\mathcal{S}$ cannot efficiently rule out length $t$ proofs in $\mathcal{T}$ of $0{=}1$ (that is,  $\mathcal{S}\,{\centernot{\sststile{}{t^{\mathcal{O}(1)}}}}\mathcal{T}\,{\centernot{\sststile{}{t}}}`0{=}1$'). 

(iii) For the $\mathcal{S}$ and $\mathcal{T}$ in (ii) and for any $\phi$ unprovable in $\mathcal{T}$, $\mathcal{S}$ cannot efficiently rule out length $t$ proofs in $\mathcal{T}$ of $\phi$ (that is, $\mathcal{T}\,{\centernot{\sststile{}{}}}\phi$ implies $\mathcal{S}\,{\centernot{\sststile{}{t^{\mathcal{O}(1)}}}}\mathcal{T}\,{\centernot{\sststile{}{t}}}\phi$).\footnote{This conjecture was formulated by the author and proved independently by Pudlák in general and by the author for sentences $x{\in}R$.}
\end{theorem}
\begin{proof}
(i) and (ii) are equivalent by Kraj\'{\i}\v{c}ek and Pudlák\cite{Krajicek}.

(ii)$\rightarrow$(iii) Suppose (ii) holds. If $\mathcal{T}$ lacks a length $t$ proof of $\phi$, there can be no proof of $0{=}1$ slightly shorter than $t$, as that would yield a length $t$ proof by contradiction of $\phi$. Therefore, if $\mathcal{S}$ efficiently proves that $\mathcal{T}$ lacks a length $t$ proof of $\phi$, contrary to (iii), it also efficiently proves that there can be no proof of $0{=}1$ slightly shorter than $t$. This contradicts (ii). Therefore, (ii) implies (iii). If $\phi$ is provable in $\mathcal{T}$, it is provable within some length $t$, so (iii) cannot hold, as $\mathcal{S}$ is consistent and cannot show $\mathcal{T}$ lacks a proof of that length.

(iii)$\rightarrow$(i) Chen et al\cite{ChenFlumMuller}.
\end{proof}

In the notation above in parentheses, write $\mathcal{T}\,{\sststile{}{}}\phi$ or $\mathcal{T}\,{\centernot{\sststile{}{}}}\phi$ respectively if $\mathcal{T}$ does or does not have a proof of $\phi$ of any length respectively. Write $\mathcal{T}\,{\sststile{}{t}}\phi$ if theory $\mathcal{T}$ has a length $t$ (or shorter) proof of sentence $\phi$ and $\mathcal{T}\,{\centernot{\sststile{}{t}}}\phi$ if not, where proof length is the number of symbols in the binary string representing the proof.\footnote{See Pudlák\cite{PudlakLengthsOfProofs}'s survey on proof length.} Likewise, $\mathcal{T}\,{\centernot{\sststile{}{t^{\mathcal{O}(1)}}}}\phi$ signifies that $\mathcal{T}$ does not have an efficient (polynomially bounded) proof of $\phi$. $\mathcal{T}\,{\sststile{}{}}\phi$ and $\mathcal{T}\,{{\sststile{}{\mathcal{O}(1)}}}\phi$ are equivalent; a provable sentence has a finite proof and is therefore provable within a constant bound.

Below, we will show that the nonexistence of an optimal proof system implies various complexity theory conjectures, by identifying some set of unprovable sentences and invoking Theorem \ref{theoremmain}(iii). In many cases, we choose unprovable sentences stating that a string $x$ is Kolmogorov random (written $x{\in}R$), that is, $x$ is incompressible by half, with no short description in the form of a program that prints $x$.\footnote{The definition in terms of incompressibility by half is arbitrary, except for Theorem \ref{thmnpintermediate} which requires logarithmic incompressibility.} Chaitin's Incompleteness Theorem states that proving $x{\in}R$ is typically impossible in a theory with a computably enumerable (c.e.) set of theorems. Otherwise, ``the first length $n$ string that provably has no short description'' would itself be a short description of some string, which is a contradiction. Here, $x{\in}R$ is an arithmetic sentence encoding that a string $x$ (represented as a binary number) lacks a short description.  Because, $R$ is dense and noncomputable, the set of $x{\in}R$ provides a dense nonconstructive pool of unprovable sentences.\footnote{For an overview of Kolmogorov complexity, see Li and Vitanyi\cite{LiVitanyiBook}. There is a rapidly growing recent literature on meta-complexity; see Santhanam\cite{SanthanamMetaComplexity}.}

More formally, define the set of Kolmogorov random strings as $R{=}\{x|\forall p{:}$ if $|p|{\leq}|x|/2$, then $p{\nearrow}$ or $p{\downarrow}$ with $U(p){\neq}x\}$, with $U$ a deterministic universal TM with no limit on its running time (not necessarily prefix free), $x$ and $p$ binary strings with $|x|$ denoting $x$'s length, $p{\downarrow}$ and $p{\nearrow}$ signifying program $p$ does or does not halt, and `$x{\in}R$' is an arithmetic sentence encoding $x{\in}R$. Single and double quotes signify a sentence, a sequence of symbols, encoding a mathematical statement.

If there is no optimal proof system, there are strong implications by Theorem \ref{theoremmain}:
\begin{itemize}\setlength\itemsep{-0.4em}
\item Ruling out length $t$ proofs is hard on with positive density, because unprovable sentences $x{\in}R$ have positive density. Equivalently, proving tautologies encoding ``there is no length $t$ proof of $x{\in}R$'' is hard with positive density. There is no optimal proof system for tautologies, with dense set of hard $\textbf{P}$-uniform families witnessing the nonoptimality.

\item A natural language is $\textbf{NP}$-intermediate: the sparse complement of the language ``$x{\in}R$ lacks a length $t$ proof'' (where $R$ is redefined, by requiring logarithmic incompressibility, to have a sparse complement). This language is not in $\textbf{P}$ but has $\textbf{P/poly}$ circuits.

\item The implicit mapping from unprovable to hard-to-prove sentences is an isomorphism. However, it is incomplete---for instance, stronger conjectures are required to imply that the polynomial hierarchy ($\textbf{PH}$) does not collapse---and substantial work may be needed to identify conjectures related to other open complexity questions and the associated isomorphisms.
\end{itemize}

The paper is organized as follows. Section \ref{preliminariessection} provides preliminaries. Section \ref{sectioncaludejurgensen} shows that unprovable sentences `$x{\in}R$' are dense among length $n$ sentences. Section \ref{sectiontautologies} discusses implications for tautologies and proof systems. Section \ref{sectionintermediate} shows that a natural language is $\textbf{NP}$-intermediate. Section \ref{sectionisomorphism} shows that the mapping from unprovable to hard-to-prove sentences is an isomorphism and discusses open questions. Section \ref{conclusion} concludes.

\section{Preliminaries}\label{preliminariessection}
\emph{Strings}: With a binary alphabet $\{0,1\}$, let $S^n$ be the set of length $n$ strings, which are ordered $n$-tuples. Let $|x|$ be the length of a string and $|S|$ be the cardinality of set $S$. A language $L$ is a subset of $\cup_{n\geq 0}S^n$.

\emph{Density}: Say the share of length $n$ strings in $L$ is bounded above zero if there exists $c>0$ such that $|L\cap S^n|/n\geq c$  for sufficiently large $n$. This implies the weaker condition that $L$ has positive upper density, i.e., that $\displaystyle\limsup_{n \rightarrow\infty}\frac{|L\cap\{1,2,\ldots,n\}|}{n}>0$. If an event depending on $n$ occurs with probability that tends to one as $n$ tends to infinity, such as $x{\in}R$ where $|x|{=}n$, say that it occurs with high probability (w.h.p.).

\emph{Theories}: Theories are assumed to be the Peano arithmetic (PA) or an extension of PA.\footnote{The conjecture could coherently refer to a weaker theory such as Robinson's Q without induction or unbounded quantifiers, which can still prove `$p{\downarrow}$' if in fact $p{\downarrow}$, by verifying the transcript of a halting computation.} To allow for average-case analysis, the standard definition of PA is modified so binary strings are encoded in arithmetic sentences as natural numbers, in binary not unary, adding a leading ``1'' to avoid losing leading zeros.

\emph{Proof Systems}: A propositional proof system is a polynomial time function $h\in \textbf{FP}$ with range $\texttt{TAUT}$ (Cook and Reckhow\cite{CookReckhow}). For tautology $\tau$, any string $w$ such that $h(w)=\tau$ is a proof of $\tau$. The proof system $h$ is \emph{optimal} if there exists $c\geq 1$ such that the length of minimal $f$ proofs of $x$ are polynomially bounded in $|x|$ with exponent $c$ by minimal $h$ proofs (Kraj\'{\i}\v{c}ek and Pudlák\cite{Krajicek}). A proof system is not optimal if and only if there is a $\textbf{P}$-uniform family of tautologies for which it requires superpolynomial proof length.

\section{Density of Unprovable Sentences}\label{sectioncaludejurgensen}
Calude and J{\"u}rgensen\cite{CaludeJurgensen} show that the share of length $n$ arithmetic sentences that are true and unprovable is bounded above zero. The result relies on two facts: `$x{\in}R$' is typically unprovable, and length $n$ strings are in $R$ w.h.p.\footnote{See the proof of Theorem 5.2 in \cite{CaludeJurgensen}.} With that context, Theorem \ref{theoremmain} implies that a similar result holds for $\texttt{coTHEOREMS}_{\leq t}{=}$ $\{\langle \phi,1^t\rangle|\mathcal{T}\,{\centernot{\sststile{}{t}}}\phi\}$.

Chaitin’s Incompleteness Theorem states:
\begin{theorem}\label{chaitinthm}
	For every consistent, arithmetically sound theory $\mathcal{T}$ with a c.e. set of theorems, $\exists k\forall x{:}|x|{>}k$, $\mathcal{T}\,{\centernot{\sststile{}{}}}`x{\in}R$'.
\end{theorem}
\begin{proof}
Otherwise, a string $x$ could be concisely described as ``the first string $x$ of length $n$ such that $\mathcal{T}$ proves `$x{\in}R$''', contrary to the definition of $R$. A TM with input $n$ in binary (of length $\log n$) could enumerate the theorems of $\mathcal{T}$, printing the first string $x$ such that $\mathcal{T}$ proves `$x{\in}R$'. Then, $k$ is determined by the length of the description of this TM, which would need to be doubled as $R$ consists of strings not compressible by half. See Li and Vitanyi\cite{LiVitanyiBook} Corollary 2.7.2 for a formal treatment.
\end{proof}

\begin{lemma}\label{lemmahighprobabilityR}
$x{\in}R$ w.h.p.
\end{lemma}
\begin{proof}
By a counting argument, the number of possible short descriptions is small. The number of length $n$ strings is $2^n$. The number of programs $p$ with $|p|{\leq}n/2$ is $2^{n/2+1}-1$, which is an upper bound on the number of length $n$ strings not in $R$. Therefore, $R$'s share of length $n$ strings is at least $1-2^{-n/2}$, so $x{\in}R$ w.h.p.
\end{proof}

Calude and J{\"u}rgensen's result implies:
\begin{theorem}\label{CaludeJurgensenInformal}
For every theory $\mathcal{T}$, the share of sentences $\{`x{\in}R$'$\,|$ $x{\in}R$ and $\mathcal{T}\,{\centernot{\sststile{}{}}}`x{\in}R$'$\}$ in length $n$ arithmetic sentences is bounded above zero, for $n$ sufficiently large. In an enumeration of sentences, for instance in lexicographic order, unprovable sentences have positive upper density.
\end{theorem}
\begin{proof}
Theory $\mathcal{T}$ cannot typically prove sentences `$x{\in}R$' where $x{\in}R$, by Theorem \ref{chaitinthm}.
The sentences `$x{\in}R$' satisfy $|`x{\in}R$'$|=|x|+c$, where $c$ is a constant not depending on $|x|$, giving the overhead of encoding `$x{\in}R$' net of $|x|$. The share of length $n$ sentences of form `$x{\in}R$' is exactly $2^{-c}$ and these satisfy $x{\in}R$ w.h.p. Therefore, for $\epsilon{>}0$, this share is bounded below by $2^{-c}{-}\epsilon$ for $n$ sufficiently large. Therefore, in an enumeration of sentences, unprovable sentences have positive upper density.
\end{proof}
The fact that a sentence `$x{\in}R$' needs only a constant $c$ bits of overhead, net of $|x|$, to encode $x{\in}R$ is needed in the next section.
\section{Tautologies and Proof Systems}\label{sectiontautologies}

A tautology can encode the sentence $\mathcal{T}{\centernot{\sststile{}{t}}}`x{\in}R$' as follows. For a given $x$, $\mathcal{T}{\centernot{\sststile{}{t}}}`x{\in}R$' is equivalent to $\langle `x{\in}R$'$,1^t\rangle{\in}\texttt{coTHEOREMS}_{\leq t}$. $\texttt{coTHEOREMS}_{\leq t}$ and $\texttt{TAUT}$ are both  $\textbf{coNP}$-complete languages, so some polynomial-time reduction $r$ from $\texttt{coTHEOREMS}_{\leq t}$ to $\texttt{TAUT}$ maps $\langle \phi,1^t\rangle$ to tautology $r(\langle \phi,1^t\rangle)$.  

Tautologies produced by the reduction $r$ confirm that every possible proof of $\mathcal{T}{\centernot{\sststile{}{t}}}`x{\in}R$' is not a valid proof. The reduction $r$ translates a family of sentences stating that no length $t$ proof exists to a family of tautologies. It should not be confused with propositional translations, which translate sentences with a single universal bounded quantifier that are easy to prove in a weak fragment of arithmetic into easy-to-prove tautologies.\footnote{See Kraj\'{\i}\v{c}ek\cite{Krajicekproof} and Cook and Nguyen\cite{CookNguyen}.}

With this encoding, two implications immediately follow: families of tautologies that are hard to prove have positive upper density in an enumeration of families, and there are dense witnesses to the nonoptimality of proof systems. 

\subsection{Proving Tautologies is Hard with Positive Density}
$R$'s density immediately implies families of tautologies hard to prove have positive upper density in an enumeration of such families. Consider an enumeration of families of Boolean formulas encoding ``no length $t$ proof of $\phi$ exists'', with each family for $\phi$ indexed by $t$, with families enumerated in lexicographic order by $\phi$. Some formulas will not be tautologies, when $\phi$ is provable within length $t$. In this enumeration, families with $\phi$ of the form `$x{\in}R$' where $x{\in}R$ have positive upper density, and these families are typically hard-to-prove tautologies.

This definition does not necessarily imply that length $n$ elements of $\texttt{TAUT}$ are average-case hard to accept. For instance, an algorithm allowed to make errors with small probability can accept for any $\phi$ of the form `$x{\in}R$' and be correct w.h.p. An error-free probabilistic polynomial time algorithm would necessarily fail with non-zero probability. 
\subsection{Dense Witnesses to Nonoptimality}
If there is no optimal proof system, then for any proof system $P$, there is a dense set of hard families of tautologies $r(\langle `x{\in}R$'$,1^t\rangle)$ letting $x$ range over all $x{\in}R$. A probabilistic, polynomial-time computable procedure to produce such a family w.h.p. is to choose a sufficiently long random string $x$. Then, $x{\in}R$ w.h.p. by Lemma \ref{lemmahighprobabilityR}, so tautologies $r(\langle `x{\in}R$'$,1^t\rangle)$ are hard for $P$ w.h.p. Tautologies that are hard for ZFC to prove are also hard for any other known proof system, as their soundness is proved by ZFC. ``Sufficiently long'' is the same as $k$ in Chaitin's theorem, based on the length of the description of a TM that enumerates the theorems of a theory. 
\section{From Turing Intermediate to $\textbf{NP}$ Intermediate}\label{sectionintermediate}
The set $R$ is Turing intermediate---it is not computable, and its complement is c.e. but not complete under many-one computable reductions (Rogers\cite{Rogers} Theorem 8.I(a) and (c)). This raises the question whether Theorem \ref{theoremmain} implies that some related language is $\textbf{NP}$-intermediate---that is, in $\textbf{NP}$, not in $\textbf{P}$, and not $\textbf{NP}$-complete under polynomial time many-one reductions. The final paragraph provides context on $\textbf{NP}$-intermediate languages.

We show that deciding the language ``has no proof of `$x{\in}R$' within length $t$'' is $\textbf{NP}$-intermediate relaxing $R$'s definition to make its complement sparse. This relaxed definition counts strings as random unless they can be compressed exponentially, not just by half. This makes the set of possible short descriptions sparse, growing polynomially in $|x|$, so the the set of non-random strings is also sparse. Define this sparse version of $R$ as $R^{sp}{=}\{x|\forall p{:}$ if $|p|{\leq}\log|x|$, then $p{\nearrow}$ or $p{\downarrow}$ with $U(p){\neq}x\}$. $R^{sp}$, like $R$, is noncomputable. Chaitin's Theorem still holds, but the parameter $k$ is exponentially larger. Fix $\mathcal{S}$ and $\mathcal{T}$ per Theorem \ref{theoremmain}.  $\mathcal{T}{\centernot{\sststile{}{t}}}`x{\in}R^{sp}$' iff $\langle `x{\in}R^{sp}$'$,1^t\rangle{\in}\texttt{coTHEOREMS}_{\leq t}$, by definition. Let $R^{sp}_t{=}\{\langle `x{\in}R^{sp}$'$,1^t\rangle|$$\mathcal{T}{\centernot{\sststile{}{t}}}`x{\in}R^{sp}$'$\}$, so $R^{sp}_t{\in}\texttt{coTHEOREMS}_{\leq t}$. Define $\overline{R^{sp}_t}{=}\{\langle `x{\in}R^{sp}$'$,1^t\rangle|\neg$$\mathcal{T}{\centernot{\sststile{}{t}}}`x{\in}R^{sp}$'$\}$. Based on $x$, $\overline{R^{sp}_t}$ can be divided into $x{\notin}R^{sp}$ where $\langle `x{\in}R^{sp}$'$,1^t\rangle{\in}\overline{R^{sp}_t}$ for all $t$, and $x{\in}R^{sp}$ where $\langle `x{\in}R^{sp}$'$,1^t\rangle{\in}\overline{R^{sp}_t}$ for sufficiently large $t$. Then:

\begin{theorem}\label{thmnpintermediate}
If there is no optimal proof system, then: (i) $\overline{R^{sp}_t}$ is $\textbf{NP}$-intermediate; and (ii) $\overline{R^{sp}_t}$ and therefore $R^{sp}_t$ have minimal circuits in $\textbf{P/poly}$ which are not $\textbf{P}$-uniform.
\end{theorem}
\begin{proof}
(i) $R^{sp}{\notin}\textbf{P}$ by assumption and Theorem \ref{theoremmain}. $\overline{R^{sp}_t}$ is sparse, as $R^{sp}$ was defined to ensure this. A sparse language is not $\textbf{NP}$-complete under many-one reductions unless $\textbf{P}{=}\textbf{NP}$, which the assumption rules out (Mahaney\cite{Mahaney82}).

(ii) $\overline{R^{sp}_t}$ is sparse, so it has minimal circuits in $\textbf{P/poly}$. These are not $\textbf{P}$-uniform, which would imply $R^{sp}{\in}\textbf{P}$, which does not hold by assumption.
\end{proof}

Ladner\cite{Ladner}  constructed artificial  $\textbf{NP}$-intermediate languages, assuming $\textbf{P}{\neq}\textbf{NP}$. Mahaney showed that a sparse language is not $\textbf{NP}$-complete under many-one reductions unless $\textbf{P}{=}\textbf{NP}$, and under Turing reductions unless $\textbf{PH}$ collapses at the second level. Ogiwara and Watanabe\cite{OgiwaraWatanabeSparseNP} provide a result employing bounded truth table reductions. Homer and Longpr{\'{e}}\cite{HomerLongpreSparse} provide additional results and alternative proofs. 

Allender and Hirahara\cite{AllenderHiraharaNPIntermediate} also provide examples of natural languages that are conditionally $\textbf{NP}$-intermediate. They show that if one-way functions exist, then approximating minimum circuit size and time-bounded Kolmogorov complexity are $\textbf{NP}$-intermediate. Determining whether these problems without approximation are $\textbf{NP}$-hard or not is an area of active research; see for instance Hirahara\cite{HiraharaMCSPNPComplete}.

If no optimal proof system exists, then $\textbf{NEXP}{\neq}\textbf{coNEXP}$ (Kraj\'{\i}\v{c}ek and Pudlák\cite{Krajicek}), and therefore there are sparse languages in $\textbf{NP}$ but not in $\textbf{P}$ (Hartmanis et al\cite{Hartmanissparse}). Our example differs by providing an explicit natural language.

\section{Isomorphisms and Open Questions}\label{sectionisomorphism}
If there is no optimal proof system, there is an implicit mapping from unprovable sentences $\phi$ to families of hard-to-prove sentences ``no length $t$ proof exists of $\phi$''. This mapping can be extended to map provable sentences to families of sentences with a length $t$ proof. If this mapping were onto, it would be an isomorphism. This is an elegant picture---an unnoticed symmetry within mathematics. However, there are several loose ends.

First, the mapping is not onto within the set of all families of  hard-to-prove sentences. Suppose theory $\mathcal{S}$ cannot efficiently prove some family of sentences not of the form ``no length $t$ proof of $\phi$ in $\mathcal{T}$ exists'' and that this family is $\textbf{P}$-uniform. We can make the mapping onto as follows. For each such family hard for $\mathcal{S}$ not in the range of the mapping, there is a sentence unprovable in $\mathcal{S}$ which states ``$\mathcal{S}$ cannot efficiently prove the family``. This is unprovable since $\mathcal{S}$ is consistent by assumption, and $\mathcal{S}$ cannot prove that it has a hard family, as it would prove its own consistency. Therefore, map this unprovable sentence onto the hard family. This extended mapping is onto. A similar solution can address the fact that a mapping from unprovable sentences to families of tautologies encoding ``no length $t$ proof exists'' is not onto.\footnote{Suppose the $\textbf{P}$-uniform family of tautologies $\tau_n$ is hard for proof system $P$ proven sound by theory $\mathcal{S}$ such that the family $\tau_n$ is also hard for $\mathcal{S}$. Then there unprovable sentences in $\mathcal{S}$: ``$\mathcal{S}$ cannot efficiently prove $\tau_n$'' and ``$P$ cannot efficiently prove $\tau_n$''.} A curious interpretation is that the role of hard families of tautologies in proof complexity, with a powerful theory such as ZFC as a proof system, can be fully understood by focusing solely on the role of unprovable sentences in ZFC. Thus, one can understand proof complexity without reference to tautologies.

Second, additional conjectures are needed to extend this question to other open questions. For instance, the conjecture ``no optimal proof system exists $\texttt{TAUT}$'', a $\Pi_1^p$-complete language, and is not strong enough to imply that $\textbf{PH}$ does not collapse. The stronger conjecture ``no optimal proof system exists for a $\Pi_2^p$-complete language, even for a proof system with an oracle for $\texttt{TAUT}$'' implies that $\Pi_2^p{\neq}\Pi_1^p$.\footnote{Chen et al\cite{ChenFlumMuller} show that a $\Pi_2^p$-complete language does not have an optimal proof system if and only if $\texttt{TAUT}$ does not have an optimal proof system, so the reference to an oracle is necessary to separate $\Pi_2^p$ and $\Pi_1^p$.} A version of Theorem \ref{theoremmain}(iii) would hold for $\mathcal{S}$ with a predicate for membership in $\Pi_1$ in the arithmetic hierarchy ($\textbf{AH}$), setting up an isomorphism for sentences with a higher degree of unsolvability.\footnote{See Pudlák\cite{PudlakLogicalFoundations} p. 569 for the construction for $\texttt{TAUT}$.} A set of such conjectures for each level of $\textbf{PH}$ would assert: $\textbf{PH}$ does not collapse due to the existence of unprovable sentences at each level of $\textbf{AH}$. These would assert, elegantly, that $\textbf{PH}$ does not collapse because $\textbf{AH}$ does not collapse.

This suggests a research program could identify a conjecture and implied isomorphism associated with each open question in complexity theory, or identify obstacles to doing so. For instance, the recent flurry of results by Liu and Pass\cite{LiuPass} and others suggest that asserting the hardness of showing $\mathcal{T}{\centernot{\sststile{}{t}}}`x{\in}R^t$', where $R^t$ is defined with respect to time-bounded Kolmogorov complexity, would imply the existence of one-way functions by asserting the average-case hardness of time-bounded Kolmogorov complexity.

To the extent each of these conjectures has the same structure, they can be rolled up into a single overarching conjecture, potentially providing insight into multiple open questions. A very strong conjecture is that some condition of the form in Theorem \ref{theoremmain}(iii) asserts the resolution of most open questions in complexity theory.
\section{Conclusion}\label{conclusion}
The conditions in Theorem \ref{theoremmain} have such strong implications for complexity theory, determining whether they are true and even provable would be desirable. An informal argument is: ruling out length $t$ proofs of an unprovable sentence $\phi$ is hard because the crucial fact is inaccessible that no proof exists of any length. This informal argument seems strongest for sentences $x{\in}R$, which are dense, nonconstructive, and typically impossible to prove. To state this in the most extreme form, suppose no other effectively computable fact about $x{\in}R$ may be useful at all in ruling out length $t$ proofs. In any program ruling out length $t$ proofs of $x{\in}R$, an if-then statement would need to compute a predicate on $R$ to determine which branch to take. Likewise, in any proof doing the same, any case-based reasoning would need to compute a predicate on $R$. However, predicates on $R$ are constrained by the fact that $R$ passes all known and conceivable effective tests of randomness (Li and Vitanyi\cite{LiVitanyiBook} Section 2.4). It is possible that if-then statements and case-based proofs might appear to behave in a purely random manner in ruling out length $t$ proofs of $x{\in}R$. If so, a program or proof can do no better than loops that exhaustively check all cases. 

This constraint might also bind non-uniformly. Boolean circuits and propositional proofs require NOT gates and cancellation to implement conditional logic, such as encoding if-then statements and case-based reasoning. Such circuits and proofs may therefore gain limited benefit their use of NOT gates and cancellation, in line with an old conjecture. It is known that for some monotone Boolean functions, the gap between their non-monotone and monotone circuit complexity (the number of gates in minimal circuits with and without NOT gates respectively) is exponential (Razborov\cite{Razborov1985}, Tardos\cite{Tardos}), and hoped that it is small for some other monotone Boolean functions such as CLIQUE (Razborov\cite{Razborov85}, Alon and Boppana\cite{AlonBoppana}). This conjecture generalized to non-monotone Boolean functions is that for certain functions, the gap is small between their cancellative and non-cancellative circuit complexity is small, where a non-cancellative circuit has a formal polynomial in which no monomial includes both a literal and its negation (Sengupta and Venkateswaran\cite{Sengupta}).\footnote{Shannon's counting argument shows that most Boolean functions require $2^n/n$ gates, the gap between cancellative and non-cancellative circuits for a random Boolean functions cannot be so large as to reduce circuits to polynomial size, as with Tardos' example.} This argument might support a claim that computational tasks such as decryption of small messages are hard in practice and not just asymptotically.
\bibliographystyle{amsplain}
%%input "equivalence.bib"
\bibliography{equivalence}

\providecommand{\bysame}{\leavevmode\hbox to3em{\hrulefill}\thinspace}
\providecommand{\MR}{\relax\ifhmode\unskip\space\fi MR }
\providecommand{\MRhref}[2]{%
  \href{http://www.ams.org/mathscinet-getitem?mr=#1}{#2}
}
\providecommand{\href}[2]{#2}
\begin{thebibliography}{10}

\bibitem{AllenderHiraharaNPIntermediate}
Eric Allender and Shuichi Hirahara, \emph{New insights on the (non-)hardness of
  circuit minimization and related problems}, {ACM} Trans. Comput. Theory
  \textbf{11} (2019), no.~4, 27:1--27:27.

\bibitem{AlonBoppana}
Noga Alon and Ravi Boppana, \emph{The monotone circuit complexity of {B}oolean
  functions}, Combinatorica \textbf{7} (1987), 1--22.

\bibitem{CaludeJurgensen}
Cristian~S. Calude and Helmut J{\"u}rgensen, \emph{Is complexity a source of
  incompleteness?}, Advances in Applied Mathematics \textbf{35} (2005), no.~1,
  1--15.

\bibitem{ChenFlumMuller}
Yijia Chen, J{\"o}rg Flum, and Moritz M{\"u}ller, \emph{Hard instances of
  algorithms and proof systems}, How the World Computes (Berlin, Heidelberg)
  (S.~Barry Cooper, Anuj Dawar, and Benedikt L{\"o}we, eds.), Springer Berlin
  Heidelberg, 2012, pp.~118--128.

\bibitem{CookNguyen}
Stephen Cook and Phuong Nguyen, \emph{Foundations of proof complexity: Bounded
  arithmetic and propositional translations}, Cambridge University Press, 2014.

\bibitem{CookReckhow}
Stephen Cook and Robert Reckhow, \emph{The relative efficiency of propositional
  proof systems}, J. Symb. Log. \textbf{44} (1979), 36--50.

\bibitem{Hartmanissparse}
Juris Hartmanis, Neil Immerman, and Vivian Sewelson, \emph{Sparse sets in
  {NP-P:} {EXPTIME} versus {NEXPTIME}}, Inf. Control. \textbf{65} (1985),
  no.~2/3, 158--181.

\bibitem{HiraharaMCSPNPComplete}
Shuichi Hirahara, \emph{{NP}-hardness of learning programs and partial {MCSP}},
  Electron. Colloquium Comput. Complex. \textbf{{TR22-119}} (2022).

\bibitem{HomerLongpreSparse}
Steven Homer and Luc Longpr{\'{e}}, \emph{On reductions of {NP} sets to sparse
  sets}, J. Comput. Syst. Sci. \textbf{48} (1994), no.~2, 324--336.

\bibitem{Rogers}
Hartley~Rogers Jr., \emph{Theory of recursive functions and effective
  computability}, MIT Press, Cambridge, MA, 1987.

\bibitem{Krajicekproof}
Jan Kraj\'{\i}\v{c}ek, \emph{Proof complexity}, Cambridge University Press, New
  York, NY, 2019.

\bibitem{Krajicek}
Jan Kraj\'{\i}\v{c}ek and Pavel Pudl{\'a}k, \emph{Propositional proof systems,
  the consistency of first order theories and the complexity of computations},
  J.~Symb. Log. \textbf{54} (1989), 1063--79.

\bibitem{Ladner}
Richard~E. Ladner, \emph{On the structure of polynomial time reducibility}, J.
  {ACM} \textbf{22} (1975), no.~1, 155--171.

\bibitem{LiVitanyiBook}
Ming Li and Paul M.~B. Vit{\'{a}}nyi, \emph{An introduction to {K}olmogorov
  complexity and its applications}, Texts in Computer Science, Springer, 2008.

\bibitem{LiuPass}
Yanyi Liu and Rafael Pass, \emph{On one-way functions and {K}olmogorov
  complexity}, 2020 IEEE 61st Annual Symposium on Foundations of Computer
  Science (FOCS), 2020, pp.~1243--1254.

\bibitem{Mahaney82}
Stephen~R. Mahaney, \emph{Sparse complete sets of {NP:} solution of a
  conjecture of {B}erman and {H}artmanis}, J. Comput. Syst. Sci. \textbf{25}
  (1982), no.~2, 130--143.

\bibitem{MonroeTCS}
Hunter Monroe, \emph{Speedup for natural problems and noncomputability},
  Theoretical Computer Science \textbf{412} (2011), no.~4-5, 478--481.

\bibitem{OgiwaraWatanabeSparseNP}
Mitsunori Ogiwara and Osamu Watanabe, \emph{On polynomial-time bounded
  truth-table reducibility of {NP} sets to sparse sets}, {SIAM} J. Comput.
  \textbf{20} (1991), no.~3, 471--483.

\bibitem{Pudlak1986length}
Pavel Pudl{\'a}k, \emph{On the length of proofs of finitistic consistency
  statements in first order theories}, Studies in Logic and the Foundations of
  Mathematics, vol. 120, Elsevier, 1986, pp.~165--196.

\bibitem{PudlakLengthsOfProofs}
\bysame, \emph{The lengths of proofs}, Handbook of Proof Theory (Samuel~R.
  Buss, ed.), Elsevier, 1998.

\bibitem{PudlakLogicalFoundations}
\bysame, \emph{Logical foundations of mathematics and computational complexity:
  A gentle introduction}, Springer, 2013.

\bibitem{PudlakFiniteDomain}
\bysame, \emph{Incompleteness in the finite domain}, Bull. Symb. Log.
  \textbf{23} (2017), no.~4, 405--441.

\bibitem{Razborov1985}
Alexander~A. Razborov, \emph{A lower bound on the monotone network complexity
  of the logical permanent}, Mathematical Notes of the Academy of Sciences of
  the USSR \textbf{37} (1985), 485--493.

\bibitem{Razborov85}
\bysame, \emph{Lower bounds on the monotone complexity of some {B}oolean
  functions}, Doklady Akademii Nauk SSSR \textbf{281} (1985), 798--801, In
  Russian. English translation in {\em Soviet Mathematics Doklady}, 31:354--57,
  1985.

\bibitem{SanthanamMetaComplexity}
Rahul Santhanam, \emph{Introduction to meta-complexity}, Meta-Complexity Boot
  Camp, {Simons Institute for the Theory of Computing}, January 2023.

\bibitem{Sengupta}
Rimli Sengupta and H.~Venkateswaran, \emph{Non-cancellative {B}oolean circuits:
  a generalization of monotone {B}oolean circuits}, Theor. Comput. Sci.
  \textbf{237} (2000), 197--212.

\bibitem{Tardos}
{\'E}va Tardos, \emph{The gap between monotone and non-monotone circuit
  complexity is exponential}, Combinatorica \textbf{8} (1988), 141--42.

\end{thebibliography}

\end{document}